\documentclass[english,british,3p ]{llncs}
\usepackage[T1]{fontenc}
\usepackage[latin9]{inputenc}
\usepackage{babel}
\usepackage{mathrsfs}
\usepackage{amsmath}
\usepackage{amssymb}
\usepackage[numbers]{natbib}
\usepackage[unicode=true]
 {hyperref}
\usepackage{breakurl}

\makeatletter

\providecommand{\tabularnewline}{\\}

\newenvironment{lyxlist}[1]
{\begin{list}{}
{\settowidth{\labelwidth}{#1}
 \setlength{\leftmargin}{\labelwidth}
 \addtolength{\leftmargin}{\labelsep}
 }}
{\end{list}}

\usepackage[vlined,plain,noresetcount]{algorithm2e} 

\SetAlgoSkip{smallskip}
\usepackage{setspace}

\usepackage[adobe-utopia]{mathdesign}

\usepackage[margin=1in]{geometry}

\makeatother

\begin{document}

\title{Notes on the existence of solutions in the pairwise comparisons method
using the Heuristic Rating Estimation approach}

\selectlanguage{english}%

\author{Konrad Ku\l{}akowski}

\institute{AGH University of Science and Technology, \\
al. Mickiewicza 30, Kraków, Poland, \\
\href{mailto:konrad.kulakowski@agh.edu.pl}{konrad.kulakowski@agh.edu.pl}}
\maketitle
\selectlanguage{british}%
\begin{abstract}
Pairwise comparisons are a well-known method for modelling of the
subjective preferences of a decision maker. A popular implementation
of the method is based on solving an eigenvalue problem for $M$ -
the matrix of pairwise comparisons. This does not take into account
the actual values of preference. The Heuristic Rating Estimation \emph{(HRE)}
approach is a modification of this method in which allows modelling
of the reference values. To determine the relative order of preferences
is to solve a certain linear equation system defined by the matrix
$A$ and the constant term vector $b$ (both derived from $M$). 

The article explores the properties of these equation systems. In
particular, it is proven that for some small data inconsistency the
$A$ matrix is an \emph{M-matrix}, hence the equation proposed by
the \emph{HRE} approach has a unique strictly positive solution. 
\end{abstract}

\section{Introduction }

The first written evidence of the use of pairwise comparisons \emph{(PC)}
dates back to the thirteenth century \citep{Colomer2011rlfa}. After
a period of growth in the first half of the twentieth century, the
pairwise comparisons method solidified in the form of the Analytic
Hierarchy Process \emph{(AHP)} proposed by Saaty \citep{Saaty1977asmfSIMPL}.
Starting as a voting method, \emph{PC} has become a way of deciding
on the relative importance (relative utility) of concepts (alternatives),
used in decision theory \citep{Saaty1977asmfSIMPL}, economics \citep{Peterson1998evabt},
psychometrics and psychophysics \citep{Thurstone27aloc} and others.
The utility of the method has been confirmed many times by various
researchers \citep{Ho2008iahp}. The theory of paired comparison is
growing all the time. Examples of such exploration are the \emph{Rough
Set} approach \citep{Greco2011fkSIMPL}, fuzzy \emph{PC} relation
handling \citep{Mikhailov2003dpff}, incomplete \emph{PC} relation
\citep{Bozoki2010oocoSIMPL,Fedrizzi2007ipca,Koczkodaj1999mnei}, data
inconsistency reduction \citep{Koczkodaj2010odbi} and non-numerical
rankings \citep{Janicki2012oapc}. A recent contribution to the pairwise
comparisons method includes the \emph{Heuristic Rating Estimation}
\emph{(HRE)} approach \citep{Kulakowski2013ahre,Kulakowski2013hreaCoRR}
that allows the user to explicitly define a reference set of concepts,
for which the utilities (the ranking values) are known a priori. The
base heuristics used in \emph{HRE} proposes to determine the relative
values of a single non-reference concept as a weighted average of
all the other concepts. Such a proposition leads to formulation a
linear equation system defined by the matrix $A$ and the strictly
positive vector of constant terms $b$. As will be shown later, in
the most interesting cases the matrix $A$ is an \emph{M-matrix} as
defined in \citep{Plemmons1976mcnm}. The sufficient condition for
$A$ to be an \emph{M-matrix} is formulated using the notion of inconsistency
referring to the quantitative relationship between entries of the
pairwise comparisons matrix $W$. In particular it is shown that the
fully consistent \emph{PC} matrix $W$ implies that $A$ is an \emph{M-matrix}.

\section{Preliminaries }

\subsection{Pairwise comparisons method}

The input to the \emph{PC} method is $M=(m_{ij})\wedge m_{ij}\in\mathbb{R}_{+}\wedge i,j\in\{1,\ldots,n\}$
a \emph{PC} matrix that expresses a quantitative relation $R$ over
the finite set of concepts $C\overset{\textit{df}}{=}\{c_{i}\in\mathscr{C}\wedge i\in\{1,\ldots,n\}\}$
where $\mathscr{C}$ is a non empty universe of concepts and $R(c_{i},c_{j})=m_{ij}$,
$R(c_{j},c_{i})=m_{ji}$. The values $m_{ij}$ and $m_{ji}$ represent
subjective expert judgment as to the relative importance, utility
or quality indicators of the concepts $c_{i}$ and $c_{j}$. Thus,
according to the best knowledge of experts it should hold that $c_{i}=m_{ij}c_{j}$
.
\begin{definition}
\label{def:A-matrix-recip}A matrix $M$ is said to be reciprocal
if for every $i,j\in\{1,\ldots,n\}$ holds $m_{ij}=\frac{1}{m_{ji}}$,
and $M$ is said to be consistent if for every $i,j,k\in\{1,\ldots,n\}$
holds $m_{ij}\cdot m_{jk}\cdot m_{ki}=1$.
\end{definition}
Since the data in the \emph{PC} matrix represents the subjective opinions
of experts, they might be inconsistent. Hence, there may exist a triad
$m_{ij},m_{jk},m_{ki}$ of entries in $M$ for which $m_{ik}\cdot m_{kj}\neq m_{ij}$.
This leads to a situation in which the relative importance of $c_{i}$
with respect to $c_{j}$ is either $m_{ik}\cdot m_{kj}$ or $m_{ij}$.
This observation gives rise to both a priority deriving method that
transforms even an inconsistent matrix of pairwise comparisons into
a consistent priority vector, and an inconsistency index describing
to what extent the matrix $M$ is inconsistent. There are a number
of priority deriving methods and inconsistency indexes \citep{Bozoki2008osak,Ishizaka2011rotm,Brunelli2013apoi}.
For the purpose of the article the \emph{Koczkodaj's inconsistency
index} is adopted \citep{Koczkodaj1993ando}. 

Let us denote: 
\begin{equation}
\kappa_{i,j,k}\overset{\textit{df}}{=}\min\left\{ \left|1-\frac{m_{ij}}{m_{ik}m_{kj}}\right|,\left|1-\frac{m_{ik}m_{kj}}{m_{ij}}\right|\right\} \label{eq:triad_inconsistency}
\end{equation}

\begin{definition}
\label{def:Koczkodaj's-inconsistency-index}Koczkodaj's inconsistency
index $\mathscr{K}$ of $n\times n$ and ($n>2)$ reciprocal matrix
$M$ is equal to
\begin{equation}
\mathscr{K}(M)=\underset{i,j,k\in\{1,\ldots,n\}}{\max}\left\{ \kappa_{i,j,k}\right\} \label{eq:koczkod_k}
\end{equation}

where $i,j,k=1,\ldots,n$ and $i\neq j\wedge j\neq k\wedge i\neq k$. 
\end{definition}
The result of the pairwise comparisons method is a ranking - a function
that assigns values to concepts. Formally, it can be defined as follows: 
\begin{definition}
\label{def:ranking_fun}The ranking function for $C$ (the ranking
of $C$) is a function $\mu:C\rightarrow\mathbb{R}_{+}$ that assigns
to every concept from $C\subset\mathscr{C}$ a positive value from
$\mathbb{R}_{+}$. 
\end{definition}
Thus, $\mu(c)$ represents the ranking value for $c\in C$. The $\mu$
function is usually defined as a vector of weights 
\begin{equation}
\mu\overset{\textit{df}}{=}\left[\mu(c_{1}),\ldots,\mu(c_{n})\right]^{T}\label{eq:weight_vector}
\end{equation}
. According to the most popular eigenvalue based approach, proposed
by \emph{Saaty} \citep{Saaty1977asmfSIMPL}, the final ranking $\mu_{\textit{ev}}$
is determined as the principal eigenvector of the $PC$ matrix $M$,
rescaled so that the sum of all its entries is $1$, i.e. 
\begin{equation}
\mu_{\textit{ev}}=\left[\frac{\mu_{\textit{max}}(c_{1})}{s_{\textit{ev}}},\ldots,\frac{\mu_{\textit{\textit{max}}}(c_{n})}{s_{\textit{ev}}}\right]^{T}\,\,\text{and}\,\,\, s_{\textit{ev}}=\underset{i=1}{\overset{n}{\sum}}\mu_{\textit{max}}(c_{i})\label{eq:eigen_value_ranking}
\end{equation}
where $\mu_{\textit{ev}}$ - the ranking function, $\mu_{\textit{max}}$
- the principal eigenvector of $M$. A more complete overview including
other methods can be found in \citep{Bozoki2008osak,Ishizaka2011rotm}.

\subsection{Heuristic Rating Estimation approach}

In the classical pairwise comparisons method the ranking function
$\mu$ for all the concepts $c\in C$ is initially unknown. Hence
every $\mu(c)$ need to be determined by the priority deriving procedure.
In real life, however, may happen that for some concepts the priority
values are known. Sometimes decision makers have extra knowledge about
the group of elements $C_{K}\subseteq C$ that allows them to determine
$\mu(c)$ for $C_{K}$ in advance. 

For example, let $c_{1},c_{2}$ and $c_{3}$ be a goods that the company
$X$ intends to place on the market, whilst $c_{4}$ and $c_{5}$
have been available for some time in stores. In order to choose the
most profitable and promising product out of $c_{1},\ldots,c_{3}$
the company $X$ want to calculate the function $\mu$ for $c_{1},c_{2}$
and $c_{3}$. Due to some similarities between $c_{1},\ldots,c_{3}$
and the pair $c_{4},c_{5}$ the company $X$ want to include them
in the ranking treating as a reference. Of course it makes no sense
to ask experts about how profitable $c_{4}$ and $c_{5}$ are. The
values $\mu(c_{4})$ and $\mu(c_{5})$ can be easily determined based
on sales reports. 

The situation as outlined in this simple example leads to the \emph{Heuristic
Rating Estimation (HRE)} model proposed in \citep{Kulakowski2013ahre,Kulakowski2013hreaCoRR}.
The main heuristics of the \emph{HRE} model assume that the set of
concepts $C$ is composed of the unknown concepts $C_{U}=\{c_{1},\ldots,c_{k}\}$
and known (reference) concepts $C_{K}=\{c_{k+1},\ldots,c_{n}\}$.
Of course, only the values $\mu_{j}$ for $c\in C_{U}$ need to be
estimated, whilst the values $\mu(c_{i})$ for $c_{i}\in C_{K}$ are
considered to be known. The adopted heuristics assumes that for every
unknown $ $$c_{j}\in C_{U}$ the value $\mu(c_{j})$ should be estimated
as the arithmetic mean of all the other values $\mu(c_{i})$ multiplied
by factor $m_{ji}$: 
\begin{equation}
\mu(c_{j})=\frac{1}{n-1}\sum_{i=1,i\neq j}^{n}m_{ji}\mu(c_{i})\label{eq:append3_eq1}
\end{equation}

Thus, the value $\mu(c_{i})$ for each unknown concept $c_{j}\in C_{U}$
is calculated according to the following formulas:

\begin{equation}
\begin{array}{ccc}
\mu(c_{1}) & = & \frac{1}{n-1}(m_{2,1}\mu(c_{2})+\dotfill+m_{n,1}\mu(c_{n}))\\
\mu(c_{2}) & = & \frac{1}{n-1}(m_{1,2}\mu(c_{1})+m_{3,2}\mu(c_{3})+\ldots+m_{n,2}\mu(c_{n}))\\
\hdotsfor[1]{3}\\
\begin{array}{c}
\mu(c_{k})\\
\\
\end{array} & \begin{array}{c}
=\\
\\
\end{array} & \begin{array}{c}
\frac{1}{n-1}\left(m_{1,k}\mu(c_{1})+\ldots\ldots+m_{k-1,k}\mu(c_{k-1})+\right.\\
\left.+m_{k+1,k}\mu(c_{k+1})+\ldots\ldots+m_{n,k}\mu(c_{n})\right)
\end{array}
\end{array}\label{eq:append3_eq2}
\end{equation}

Since the values $\mu(c_{k+1}),\ldots,\mu(c_{n})$ are known and constant
($c_{k+1},\ldots,c_{n}$ are the reference concepts), they can be
grouped together. Let us denote:

\begin{equation}
b_{j}=\frac{1}{n-1}m_{k+1,j}\mu(c_{k+1})+\ldots+\frac{1}{n-1}m_{n,j}\mu(c_{n})\label{eq:append3_eq3}
\end{equation}

Thus (\ref{eq:append3_eq2}) could be written as the linear equation
system $A\mu=b$ where the matrix $A$ is:

\begin{equation}
A=\left[\begin{array}{ccc}
1 & \cdots & -\frac{1}{n-1}m_{1,k}\\
-\frac{1}{n-1}m_{2,1} & \cdots & -\frac{1}{n-1}m_{2,k}\\
\vdots & \ddots & \vdots\\
-\frac{1}{n-1}m_{k,1} & \cdots & 1
\end{array}\right],\label{eq:A_and_b}
\end{equation}

and the vectors $b$ and $\mu$ are: 
\begin{equation}
b=\left[\begin{array}{c}
\frac{1}{n-1}\sum_{i=k+1}^{n}m_{1,i}\mu(c_{i})\\
\frac{1}{n-1}\sum_{i=k+1}^{n}m_{2,i}\mu(c_{i})\\
\vdots\\
\frac{1}{n-1}\sum_{i=k+1}^{n}m_{k,i}\mu(c_{i})
\end{array}\right],\,\,\,\,\,\,\,\mu=\left[\begin{array}{c}
\mu(c_{1})\\
\mu(c_{2})\\
\vdots\\
\mu(c_{2})
\end{array}\right]\label{eq:vectors_b_and_mu}
\end{equation}

It is worth noting that $b>0$, since every $b_{i}$ for $i=1,\ldots,k$
is a sum of strictly positive components. According to (Def.~\ref{def:ranking_fun})
the ranking results must be strictly positive, hence only strictly
positive vectors $\mu$ are considered to be feasible.

\subsection{M-matrices}

The answer to the question concerning the existence of solution of
the linear equation system $A\mu=b$ requires knowledge of certain
properties of the \emph{M-matrix} \citep{Quarteroni2000nm}. For this
purpose, let us denote $\mathcal{M}_{\mathbb{R}}(n)$ - the set of
$n\times n$ matrices over $\mathbb{R}$, $\mathcal{M}_{\mathbb{Z}}(n)$
- the set of all $A=[a_{ij}]\in\mathcal{M}_{\mathbb{R}}(n)$ with
$a_{ij}\leq0$ if $i\neq j$ and $i,j\in\{1,\ldots,j\}$. Moreover,
for every matrix $A\in\mathcal{M}_{\mathbb{R}}(n)$ and vector $b\in\mathbb{R}^{n}$
the notation $A\geq0$ and $b\geq0$ will mean that each entry of
$A$ and $b$ is non-negative and neither $A$ nor $b$ equals $0$.
The spectral radius of $A$ is defined as $\rho(A)=\max\{|\lambda|:\det(\lambda I-A)=0\}$.
\begin{definition}
\label{def:M-matrix-def}An $n\times n$ matrix that can be expressed
in the form $A=sI-B$ where $B=[b_{ij}]$ with $b_{ij}\geq0$ for
$i,j\in\{1,\ldots,n\}$, and $s\geq\rho(B)$, the maximum of the moduli
of the eigenvalues of B, is called \emph{M-matrix}.
\end{definition}
In practice, solving many problems in the biological and social sciences
can be reduced to problems, involving \emph{M-matrices} \citep{Plemmons1976mcnm}.
For this reason, \emph{M-matrices} have been of interest to researchers
for a long time and many of their properties have already been proven.
Following \citep{Plemmons1976mcnm} some of them are recalled below
in the form of the Theorem~\ref{PlemmonsTheo}.
\begin{theorem}
\label{PlemmonsTheo}For every $A\in\mathcal{M}_{\mathbb{Z}}(n)$
each of the following conditions is equivalent to the statement: $A$
is a nonsingular \emph{M-matrix}.\end{theorem}
\begin{enumerate}
\item $A$ is inverse positive. That is, $A^{-1}$ exists and $A^{-1}\geq0$
\item $A$ is semi-positive. That is, there exists $x>0$ with $Ax>0$ 
\item There exists a positive diagonal matrix $D$ such that $AD$ has all
positive row sums.
\end{enumerate}
In the context of equation $A\mu=b$ it is worth noting that if $A$
is nonsingular then $A^{-1}$ is also nonsingular, and thus the the
vector $\mu$ could be determined as $A^{-1}b$. Moreover for $b>0$
(every entry of vector $b$ is a sum of strictly positive values)
and $A$ - \emph{M-matrix}, due to the theorem above $A^{-1}\geq0$,
the vector $\mu$ also must be strictly positive i.e. $\mu=A^{-1}b>0$.

\section{Inconsistency based condition for the existence of a solution}

The entries of $M=[m_{ij}]$ represent comparative opinions of experts,
they are thus inherently strictly positive, that is $M>0$. For the
same reason the matrix $A$ (\ref{eq:A_and_b}), formed on the basis
of $M$, has positive entries only on the diagonal, i.e. $A\in\mathcal{M}_{\mathbb{Z}}(n)$.
Therefore proving that $A$ satisfies any of the conditions of the
Theorem \ref{PlemmonsTheo}, implies that $A$ is an \emph{M-matrix}.
Hence, due to the remarks below the Theorem \ref{PlemmonsTheo}, and
the fact that in the \emph{HRE} approach $b>0$, the equation $A\mu=b$
has only one strictly positive solution $\mu$. 

The sufficient condition for $A$ to be an \emph{M-matrix} is formulated
with the help of the inconsistency index $\mathcal{K}(M)$ (Def.~\ref{def:Koczkodaj's-inconsistency-index}).
Using an inconsistency index simplifies the evaluation of $A\mu=b$
and enables linking the reliability of expert assessments (the paired
ranking for which the inconsistency index is too high are considered
as unreliable \citep{Saaty1977asmfSIMPL}) with the solution existence
problem. 
\begin{theorem}
\label{main_theorem}The linear equation system $A\mu=b$ introduced
in the HRE approach has exactly one strictly positive solution for
$0<r\leq n-2$ if 
\begin{equation}
\mathcal{K}(M)<1-\frac{1+\sqrt{1+4(n-1)(n-r-2)}}{2(n-1)}\label{eq:condit}
\end{equation}

where $n=|C_{U}\cup C_{K}|$ - is the number of all the estimated
concepts, $r=|C_{K}|$ - is the number of the known concepts. \end{theorem}
\begin{proof}
Following (Def. \ref{def:Koczkodaj's-inconsistency-index}), the value
of \emph{Koczkodaj's inconsistency index} $\mathscr{K}(M)$, in short
$\mathscr{K}$, means that the maximal inconsistence for some triad
$m_{pq},m_{qr}$ and $m_{pr}$ is $\mathscr{K}$. Thus, in the case
of an arbitrarily chosen triad $m_{ik},m_{kj},m_{ij}$ it must hold
that: 
\begin{equation}
\mathscr{K}\geq\kappa_{i,j,k}=\min\left\{ \left|1-\frac{m_{ij}}{m_{ik}m_{kj}}\right|,\left|1-\frac{m_{ik}m_{kj}}{m_{ij}}\right|\right\} \label{eq:K_label}
\end{equation}
This means that either: $m_{ij}\leq m_{ik}m_{kj}$ implies that $\mathscr{K}\geq1-\frac{m_{ij}}{m_{ik}m_{kj}}$,
or $m_{ik}m_{kj}\leq m_{ij}$ implies that $\mathscr{K}\geq1-\frac{m_{ik}m_{kj}}{m_{ij}}$.
Denoting 
\begin{equation}
\alpha\overset{\textit{df}}{=}1-\mathscr{K}\label{eq:alpha_def}
\end{equation}
 we obtain the result that either $m_{ij}\leq m_{ik}m_{kj}$ implies
$m_{ij}\geq\alpha\cdot m_{ik}m_{kj}$, or $m_{ik}m_{kj}\leq m_{ij}$
implies that $\frac{1}{\alpha}\cdot m_{ik}m_{kj}\geq m_{ij}$. It
is easy to see that $0\leq\mathcal{K}<1$, thus $0<\alpha\leq1$.
Thus, both these assertions lead to the common conclusion: 

\begin{equation}
\alpha\cdot m_{ik}m_{kj}\leq m_{ij}\leq\frac{1}{\alpha}m_{ik}m_{kj}\label{eq:triad_estim}
\end{equation}

for every $i,j,k\in\{1,\ldots,n\}$. This mutual relationship between
entries of $M$ can be written as the parametric equation $m_{ij}=t\cdot m_{ik}m_{kj}$
where $\alpha\leq t\leq\frac{1}{\alpha}$. Using this equation the
matrix $A$ (see \ref{eq:A_and_b}) can be written as:
\begin{equation}
A=\left[\begin{array}{ccc}
t_{1,1}m_{1,k}m_{k,1} & \cdots & -\frac{m_{1,k}}{n-1}\\
\vdots & \vdots & \vdots\\
-\frac{t_{k-1,1}m_{k-1,k}m_{k,1}}{n-1} & \ddots & -\frac{m_{k-1,k}}{n-1}\\
-\frac{t_{k,1}m_{k,1}}{n-1} & \cdots & 1
\end{array}\right]\label{eq:A_eq}
\end{equation}

where $\alpha\leq t_{ij}\leq\frac{1}{\alpha}$, for $i,j\in\{1,\ldots,k-1\}$.
Hence, finally the matrix $A$ can be written as the matrix product
$A=BC$ where: 
\begin{equation}
B=\left[\begin{array}{cccc}
t_{1,1}m_{1,k} & \cdots & \cdots & -\frac{m_{1,k}}{n-1}\\
\vdots & \ddots & \vdots & \vdots\\
-\frac{t_{k-1,1}m_{k-1,k}}{n-1} & \vdots & t_{k-1,k-1}m_{k-1,k} & -\frac{m_{k-1,k}}{n-1}\\
-\frac{t_{k,1}}{n-1} & \cdots & -\frac{t_{k,k-1}}{n-1} & 1
\end{array}\right]\label{eq:B_eq}
\end{equation}

and

\begin{equation}
C=\left[\begin{array}{cccc}
m_{k,1} & 0 & \cdots & 0\\
0 & \ddots & \cdots & 0\\
\vdots & \vdots & m_{k,k-1} & \vdots\\
0 & \cdots & \cdots & 1
\end{array}\right]\label{eq:C_eq}
\end{equation}

Since both $t_{ij}$ and $m_{ij}$ are strictly positive, it holds
that $B\in\mathcal{M}_{\mathbb{Z}}(n)$. Therefore, due to the third
condition of the Theorem \ref{PlemmonsTheo} where $D\overset{\textit{df}}{=}I$,
$B$ is a nonsingular \emph{M-matrix} if sums of all its rows are
positive. In other words $B$ is \emph{an M-matrix }if each of the
following inequalities (\ref{eq:equations_cond}) are true. 

\begin{equation}
\begin{array}{ccc}
m_{1,k}(n-1)t_{1,1}-m_{1,k}(t_{1,2}+t_{1,3}+\dotfill+t_{1,k-1}+1) & \geq & 0\\
m_{2,k}(n-1)t_{2,2}-m_{2,k}(t_{2,1}+t_{2,3}+\dotfill+t_{2,k-1}+1) & \geq & 0\\
\hdotsfor[1]{3}\\
(n-1)-(t_{k,1}+t_{k,2}+\dotfill+t_{k,k-1}) & \geq & 0
\end{array}\label{eq:equations_cond}
\end{equation}

Due to the constraints introduced by the inconsistency $\mathcal{K}(M)$
the minimal and the maximal value of every $t_{ij}$ is $\alpha$
and $\frac{1}{\alpha}$ correspondingly. Thus the inequalities (\ref{eq:equations_cond})
are true if the following two inequalities are satisfied: 
\begin{equation}
(n-1)\alpha>(\underset{n-r-2}{\underbrace{\frac{1}{\alpha}+\ldots+\frac{1}{\alpha}}}+1)\,\,\,\,\,\text{and}\,\,\,\,\,(n-1)>(\underset{n-r-1}{\underbrace{\frac{1}{\alpha}+\ldots+\frac{1}{\alpha}}})\label{eq:cond_fin_eq1}
\end{equation}

where $r=n-k$ is the number of elements in $C_{K}$. In other words
$B$ is \emph{an M-matrix} if the following two conditions are met:
\begin{equation}
f(\alpha)>0,\,\,\,\text{where}\,\,\, f(\alpha)\overset{\textit{df}}{=}(n-1)\alpha^{2}-\alpha-(n-r-2)\label{eq:larger_eq1}
\end{equation}

and 
\begin{equation}
g(\alpha)>0,\,\,\,\text{where}\,\,\, g(\alpha)\overset{\textit{df}}{=}(n-1)\alpha-(n-r-1)\label{eq:g_alpha_eq}
\end{equation}

By solving $f(\alpha)=0$ and choosing the larger root %
\footnote{The smaller root $\frac{1-\sqrt{1+4(n-1)(n-r-2)}}{2(n-1)}\leq0$ for
any $n=3,4\ldots$ and $0<r\leq n-2$, so it does not need to be taken
into account.%
} we obtain the result that: 
\begin{equation}
\mathcal{K}(M)<1-\frac{1+\sqrt{1+4(n-1)(n-r-2)}}{2(n-1)}\label{eq:cond1}
\end{equation}

whilst the right, linear inequality $g(\alpha)>0$ leads to 
\begin{equation}
\mathcal{K}(M)<1-\frac{(n-r-1)}{(n-1)}\label{eq:cond2}
\end{equation}

In order to decide which of these criteria are more restrictive and
which should therefore be chosen, the following two cases need to
be considered: 
\begin{lyxlist}{00.00.0000}
\item [{(a)}] $r=n-2$ 
\item [{(b)}] $0<r\leq n-3$
\end{lyxlist}
When $r=n-2$ it is easy to see that $f(\alpha)=\alpha g(\alpha)$.
Thus both functions $f(\alpha)$ and $g(\alpha)$ take the $0$ value
for the same values of argument $\alpha$. Hence, both criteria (\ref{eq:cond1})
and (\ref{eq:cond2}) are equal. 

If $0<r\leq n-3$ it is easy to see%
\footnote{To demonstrate this please consider the sequence of inequalities $\left(\frac{1+\sqrt{1+4(n-1)(n-r-2)}}{2(n-1)}\right)^{2}\geq\ldots\geq\frac{4(n-1)(n-r-2)}{4(n-1)^{2}}\geq\left(\frac{n-r-1}{n-1}\right)^{2}$.%
} that the first condition (\ref{eq:cond1}) is more restrictive than
(\ref{eq:cond2}), i.e. wherever (\ref{eq:cond1}) holds the inequality
(\ref{eq:cond2}) is also true. In other words, to provide a guarantee
that $B$ is an \emph{M-matrix} it is enough to consider the more
restrictive condition (\ref{eq:cond1}). 

The fact that $B$ is \emph{an M-matrix} implies that there exists
an inverse matrix $B^{-1}\geq0$ (Theorem \ref{PlemmonsTheo}). Hence,
due to the form of the matrix $C$ it is easy to see that the inverse
matrix $C^{-1}$ exists, thus $A^{-1}$ exists and $A^{-1}=C^{-1}B^{-1}\geq0$.
Thus, due to the first condition of the Theorem \ref{PlemmonsTheo},
$A$ is \emph{an M-matrix}, which means that the equation $A\mu=b$
has a unique strictly positive solution. This conclusion completes
the proof of the theorem. 
\end{proof}
Of course, the theorem proven above does not address the case $r=n-1$.
This is because $r=n-1$ implies $A$ is a scalar, hence solving $A\mu=b$
is trivial. When $M$ is fully consistent, i.e. $\mathcal{K}(M)=0$
and $\alpha=1$, it is easy to see that both conditions (\ref{eq:cond_fin_eq1})
are satisfied. Thus, in such a case $A$ is an \emph{M-matrix}, and
what follows $A\mu=b$ always has strictly positive solution. Several
upper bounds for $\mathcal{K}(M)$ related to parameters $n$ and
$r$ arising from the above theorem are gathered in the Table \ref{tbl:maximal_K_vals}. 

\begin{table}
\begin{centering}
\begin{tabular}{|c|c|c|c|c|c|}
\hline 
$0\leq\mathcal{K}(M)<$ & $r=1$ & $r=2$ & $r=3$ & $r=4$ & $r=5$\tabularnewline
\hline 
$n=3$ & $0.5$ & - & - & - & -\tabularnewline
\hline 
$n=4$ & $0.232$ & $0.666$ & - & - & -\tabularnewline
\hline 
$n=5$ & $0.156$ & $0.359$ & $0.75$ & - & -\tabularnewline
\hline 
$n=6$ & $0.118$ & $0.259$ & $0.441$ & $0.8$ & -\tabularnewline
\hline 
$n=7$ & $0.095$ & $0.204$ & $0.333$ & $0.5$ & $0.833$\tabularnewline
\hline 
\end{tabular}
\par\end{centering}

\caption{The upper bounds for $\mathcal{K}(M)$ for which there is a guarantee
that $A$ is an \emph{M-matrix}}
\label{tbl:maximal_K_vals}
\end{table}

\begin{remark}
Let us note that for any combination of $r$ and $n$ where $0<r\leq n-2$,
where $r,n\in\mathbb{N}_{+}$ the right side of (\ref{eq:cond1})
is greater than $0$. In other words for a sufficiently low inconsistency
the equation $A\mu=b$ always has a feasible solution. 
\end{remark}
To prove this (see \ref{eq:cond1}) it is enough to show that for
$n=3,4,\ldots$ holds: 
\begin{equation}
\frac{\left(1+\sqrt{1+4(n-1)(n-r-2)}\right)}{2(n-1)}<1\label{eq:exist_eq_1}
\end{equation}
Since $\sqrt{1+4(n-1)(n-r-2)}\leq\sqrt{1+4(n-1)(n-3)}$, thus in particular
\begin{equation}
\frac{\left(1+\sqrt{1+4(n-1)(n-3)}\right)}{2(n-1)}<1\label{eq:exist_eq_2}
\end{equation}
Thus, 
\begin{equation}
\sqrt{1+4(n-1)(n-3)}<2n-3\label{eq:exist_eq_3}
\end{equation}
and 

\begin{equation}
4(n-1)(n-3)<\left(2n-3\right)^{2}-1\label{eq:exist_eq_4}
\end{equation}
which is equivalent to

\begin{equation}
4(n-1)(n-3)<4(n-1)(n-2)\label{eq:exist_eq_5}
\end{equation}
Thus, for every $n>1$ the above equation reduces to: 
\begin{equation}
n-3<n-2\label{eq:exist_eq_6}
\end{equation}
The last inequality is always satisfied, which proves that (\ref{eq:exist_eq_1})
is true for $n\geq3$. 
\begin{remark}
Another interesting observation is that the proof of the Theorem \ref{main_theorem}
takes into account only that entries of the matrix $M$, that make
up the matrix $A$. Hence, there is no need to analyse the inconsistency
for the whole matrix $M$. Instead, it is enough to analyse $\widetilde{M}$
- the minor of $M$ whose rows and columns correspond to the elements
from the set of unknown concepts $C_{U}$. It also holds that $\mathscr{K}(\widetilde{M})\leq\mathscr{K}(M)$.
Thus, it may happen that the inconsistency of $\widetilde{M}$ meets
the condition (\ref{eq:cond1}), whilst the inconsistency of $M$
is too high.

Assuming that $C_{U}=\{c_{1},\ldots,c_{k}\}$ let us define the matrix
$\widetilde{M}$ as follows:
\begin{equation}
\widetilde{M}=\left[\begin{array}{ccc}
1 & \cdots & m_{1,k}\\
\vdots & \vdots & \vdots\\
m_{k-1,1} & \ddots & m_{k-1,k}\\
m_{k,1} & \cdots & 1
\end{array}\right]\label{eq:minor_matrix}
\end{equation}
It might be noticed that assuming $\alpha\overset{\textit{df}}{=}1-\mathscr{K}(\widetilde{M})$
in (\ref{eq:alpha_def}) the proof of the Theorem \ref{main_theorem}
does not change. Hence, instead of exploring the inconsistency of
$M$ it is sufficient to examine the inconsistency of its minor $\widetilde{M}$.
Thereby, the upper bounds given in the Table \ref{tbl:maximal_K_vals}
can be applied to $\mathscr{K}(\widetilde{M})$ instead of $\mathscr{K}(M)$. 

By definition of \emph{Koczkodaj inconsistency index,} $\mathscr{K}(M)$
is the maximum of $T_{M}=\{\kappa_{i,j,r}:i,j,r\in\{1,\ldots,n\}\}$
. Similarly, $\mathscr{K}(\widetilde{M})$ is the maximum of $T_{\widetilde{M}}=\{\kappa_{i,j,r}:i,j,r\in\{1,\ldots,k\}\}$,
where $k$ is the number of elements in $C_{U}$. It is easy to see
that $T_{\widetilde{M}}\subseteq T_{M}$. This implies that also $\max T_{\widetilde{M}}\leq\max T_{M}$,
which leads to the conclusion that $\mathscr{K}(\widetilde{M})\leq\mathscr{K}(M)$. 
\end{remark}

\section{Summary }

The reliability of the results achieved in the \emph{PC} models are
inseparably linked to the degree of inconsistency of the input data
\citep{Saaty1977asmfSIMPL}. The lower the inconsistency the better
and more reliable the results might be expected to be. Therefore,
most practical applications of the \emph{PC} method seek to construct
the \emph{PC} matrix with the smallest possible inconsistency. The
theorem proven in this article is in line with the tendency to seek
\emph{PC} solutions with low inconsistency. It shows that for an appropriately
small inconsistency $\mathcal{K}(M)$ the recently proposed \emph{Heuristic
Rating Estimation} method always has a feasible solution.

\bibliographystyle{elsart-num-sort}
\bibliography{papers_biblio_reviewed}

\end{document}